\documentclass[a4paper, 10pt]{article}

\usepackage[square,sort,comma,numbers]{natbib}
\usepackage{graphicx}
\usepackage{dcolumn}
\usepackage{bm}
\usepackage{amsmath}
\usepackage{amssymb,amsthm}
\usepackage{mathrsfs}
\usepackage[bottom]{footmisc}
\usepackage{bbm}
\usepackage{bbold}
\usepackage{mathtools}
\usepackage{amsfonts}
\usepackage{epsfig,color}

\definecolor{nred}{rgb}{0.2,0.2,0.7}
\definecolor{ngreen}{rgb}{0.2,0.6,0.2}
\definecolor{nred}{rgb}{0.7,0.2,0.2}
\definecolor{nblack}{rgb}{0,0,0}

\newcommand{\ket}[1]{|#1\rangle}
\newcommand{\ketbra}[2]{\ket{#1}\!\bra{#2}}
\newcommand{\bra}[1]{\langle#1|}
\newcommand{\braket}[2]{\langle #1 | #2 \rangle}
\newcommand{\proj}[1]{\ket{#1}\!\bra{#1}}
\newcommand{\tr}{\text{tr}}

\renewcommand{\H}{\mathcal{H}}

\def\G{\mathrm{G}}

\newtheorem{theorem}{Theorem}

\newtheorem{lem}[theorem]{Lemma}

\def\g{\mathrm{guess}}

\def\N{\mathcal{N}}

\def\tr{\mbox{tr}}
\def\bea{\begin{eqnarray}}
\def\eea{\end{eqnarray}}

\def\i{\mbox{id}}

\begin{document}

\title{ Detecting Noisy Channels by Channel Discrimination : \\Local versus Entangled Resources }

\date{\today} 

\author{ Joonwoo Bae$^{1}$ 
~and ~ Tanmay Singal$^{2}$ 
 \\ \\[0.1em]
{\it\small   $^{1}$ School of Electrical Engineering, Korea Advanced Institute of Science and Technology (KAIST),} \\
{\it\small 291 Daehak-ro, Yuseong-gu, Daejeon, 34141, Korea, and }  \\
{\it\small   $^{2}$Department of Applied Mathematics, Hanyang University (ERICA),} \\
{\it\small 55 Hanyangdaehak-ro, Ansan, Gyeonggi-do, 426-791, Korea.}  }

\maketitle

\begin{abstract}
Dynamics of many-qubit systems, that may correspond to computational processing with quantum systems, can be efficiently and generally approximated by a sequence of two- and single-qubit gates. In practical applications, however, a quantum gate prepared as a unitary transformation may appear as a noisy channel and consequently may inhibit quantum advantages. In this work, we apply the scheme of channel discrimination to detect if a quantum gate that is actually realized is unitary or noisy. We show that a two-qubit unitary transformation and its noisy counterpart can be optimally discriminated by local resources, without the necessity of creating entanglement repeatedly. It is also shown that the scheme can be applied to estimation of the fraction of noise existing in quantum gates. 
\end{abstract}


\maketitle

\section{Introduction}
\label{introduction}

As quantum dynamics, that may lead to quantum advantages, can be realized by concatenation of quantum gates, it is of great importance to prepare quantum gates with a high precision for quantum information applications. Remarkably, arbitrary quantum dynamics can be constructed by composing single- and two-qubit gates only ~\cite{ref:kitaev} ~\cite{ref:deutsch} ~\cite{ref:seth}, which are then called elementary quantum gates ~\cite{ref:bennett}. It appears that realisation of single-qubit states is often feasible in experiment, while two-qubit gates and their concatenation still remain challenging. This can be equivalent to the difficulty of generating a large-size entangled state ~\cite{ref:briegel}. Once entangled states are obtained, they are a general resource for quantum information processing ~\cite{ref:ent}. Therefore, quantum interactions that lead to composition of quantum gates are also a resource for quantum information processing. 


\begin{figure}[t]
\begin{center}
\includegraphics[width=3.5in]{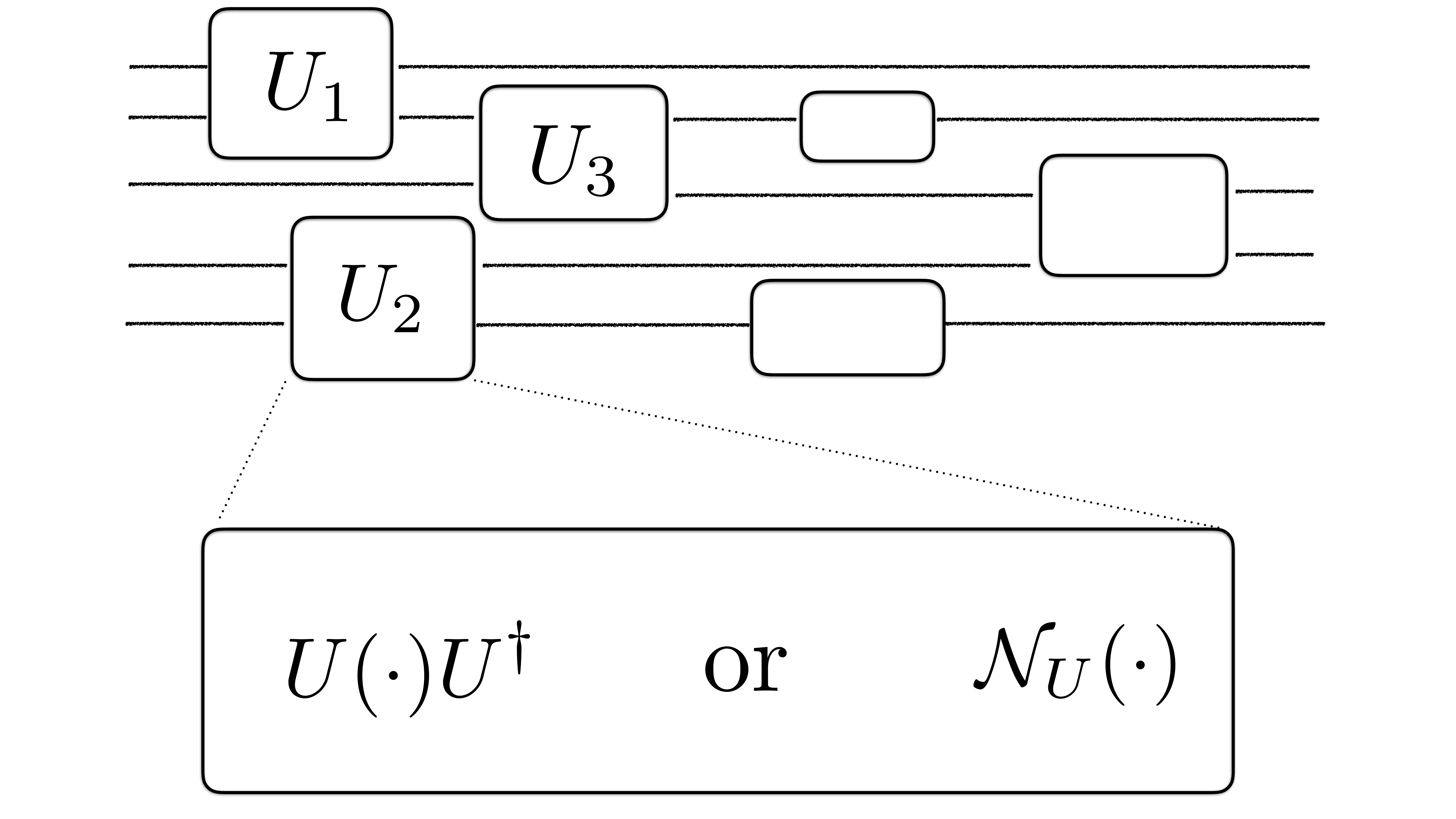}
\caption{ A unitary gate on quantum states can be decomposed into a sequence of single and two-qubit gates, $U \approx U_1 U_2 U_3 \cdots$, to arbitrary precision, where two-qubit gates introduce interactions of qubits and can generate entangled states. In a practical realization, it may happen that a gate suffers unwanted interaction with environment so that its noisy counterpart is actually implemented on qubits. For instance, the quantum gate $U_2$ that transforms a state $\rho$ as $U_2 (\rho) = U \rho U^{\dagger}$ may be actually realised as a noisy channel: $\N_{U} (\rho)$ with some probability. From a practical point of view, one would like to know which of $U (\cdot) U^{\dagger}$ and $\N_{U} (\cdot)$ has been actually realized. }\label{fig:fic}
\end{center}
\end{figure}

Let us now consider a practical scenario of realizing a two-qubit gate, denoted by $U$. There often exist unwanted interactions with environment, due to which one ends up with a noisy channel, $\N_U$, that corresponds to a completely positive and trace-preserving map over quantum states. In other words, once a two-qubit gate is implemented, it may appear as a noisy channel with some probability, see Fig. (\ref{fig:fic}). If a two-qubit gate containing some noise is incorporated to a quantum circuit, the noise may propagate over the circuit since the gate introduces interactions among quantum systems. One can devise a quantum circuit where a local noise does not contaminate other systems. If this demands more resources, such as is the case when designing fault-tolerant circuits for quantum error correcting codes, which requires even more quantum interactions and entanglement, from a practical point of view one may prefer to detect a noisy quantum gate immediately, so that the gate may be chosen individually to fix or improve. This defines the problem of channel identification, that aims to find which of $U$ and $\N_U$ is actually realized as the operation in a quantum gate. One may also have {\it a priori} probabilities about which one is to appear. To tackle the problem in a brute force way, one could apply quantum process tomography that identifies the  quantum channel that has been actually performed ~\cite{ref:qpt1} ~\cite{ref:qpt2}. This, however, has a higher cost due to a large number of quantum measurements together with classical post-processing. 

We here consider the strategy of optimal quantum channel discrimination to find if a two-qubit gate realized in a circuit is unitary or noisy, i.e., channel discrimination between a two-qubit gate $U$ and its noisy operation $\N_U$. This leads to a cost effective method fitted in a scenario of one-shot channel identification. In particular, we address the question of how useful local operations and classical communication (LOCC) is in the scenario of channel discrimination. This is because measurement in an entangled basis repeatedly requires an implementation of quantum interactions, and such interactions may again be tainted with the same noisy imperfections as happens in $\N_U$. In fact, we show that optimal channel discrimination can be achieved by local resources only, i.e., LOCC, and also provide an LOCC protocol for optimal discrimination. Moreover, we also show that the protocol can be applied for finding how noisy $\N_U$ is, that is, estimation of the noise fraction. Examples are also presented. 

The paper is structured as follows. In the next Sec. \ref{summary} we summarize the results. In sec. \ref{section:problem}, we introduce the problem in more detail, after which in Sec. \ref{section:optimaldiscrimination} we prove that a two-qubit gate and its noisy counterpart can be optimally discriminated using only LOCC. In Sec. \ref{sec:ex} relevant examples are shown, the controlled-NOT gate and the SWAP gate, in which the LOCC discrimination is described in detail. In Sec. \ref{sec:discussion}, the usefulness of the protocol for LOCC discrimination is discussed, firstly in the estimation of the noise fraction and secondly in various types of noisy channels. In Sec. \ref{sec:conclusion}, we conclude the result. 

\section{Summary of Results}
\label{summary}

Although quantum channels and quantum states are closely related to each other, e.g. ~\cite{ref:c1, ref:c2, ref:c3}, it has been recognized that they are not equivalent in tasks of distinguishability. First of all, ancillary systems can improve discrimination of quantum channels, which is not the case with quantum states. Namely, for any pair of channels $\N_{i}: S(\H_{\mathrm{in}})\rightarrow S(\H_{\mathrm{out}} )$ for $i=1,2$ where $S(\H)$ denotes the set of quantum states on a Hilbert space $\H$, there exists a quantum state $\rho\in S(\H_1 \otimes \H_A)$ that improves distinguishability of the quantum channels ~\cite{ref:piani}. This is in fact equivalent to the condition that entanglement is contained in a quantum state $\rho\in S(\H_1 \otimes \H_A)$ ~\cite{ref:piani} ~\cite{ref:baedarek}. 

For unitary transformations, it has been shown that, for two unitaries $U$ or $V$ if they can be repeatedly applied to a fixed quantum state without ancillary systems attached, there exists a finite number of repetitions $N_G$ after which the two unitary transformations $U^{N_G}$ and $V^{N_G}$ can be perfectly discriminated ~\cite{ref:acin} ~\cite{ref:renes}. The results can be extended to unitary transformations for multipartite systems, even without requirement of entangled states.  Namely, perfect discrimination of multipartite unitary transformations can be done within a finite number $N_{LOCC}$ of repetitions in an LOCC scenario ~\cite{ref:ch1} ~\cite{ref:ch2}. It is, nonetheless, noteworthy that more repetitions may be necessary in an LOCC scenario, that is, we have $N_{LOCC}\geq N_{G}$ in general ~\cite{ref:ch1} ~\cite{ref:ch2} ~\cite{ref:bae}. 

The main contribution of the present article is to investigate optimal discrimination between a two-qubit gate $U$ and its noisy counterpart $\N_U$, and to find its applications to the problem of deciding if a two-qubit gate in a quantum circuit is unitary or noisy. In particular, it is shown that optimal channel discrimination can be achieved using local resources only: a protocol for optimal channel discrimination is provided, which nowhere applies entangled resources but works only with LOCC, preparation of a separable state and local measurements with classical communication.

The channel discrimination with local resources works as follows: to establish which of $U$ or $\N_U$ is realized in the circuit, we choose our input state $\rho$ to maximize the distinguishability between $U \rho U^\dag$ and $\N_U [\rho]$. This maximization is done over all input states $\rho$. This optimization can be achieved with a pure state, denoted by $|\psi\rangle$, such that $U \ketbra{\psi}{\psi} U^\dag$ and $\N_U [ \ketbra{\psi}{\psi}]$ are rank one and rank-three states. We show that one can prepare a product input state $\ket{\psi}$ such that the output states $U \ket{\psi}$ and $\N_U[\ket{\psi}]$ can be optimally discriminated using LOCC. For this, the results are shown in the following order. In Theorem \ref{thm:1}, it is shown that optimal discrimination of two bipartite quantum states with global operations can also be achieved by LOCC if and only if two measurement operators, that is, positive-operator-valued-measures (POVMs), for the optimal discrimination can be discriminated perfectly by LOCC. This allows one to consider the equivalent formulation of LOCC distinguishability. We then show that the optimal measurement that discriminates between a two-qubit unitary transformation and its noisy counterpart contains rank-one and rank-three POVMs. Lemma \ref{thm:theory:1} shows that these POVMs can be perfectly discriminated by LOCC if and only if one of them is separable. In Theorem \ref{thm:theory:2}, it is shown that for a two-qubit unitary transformation and its noisy counterpart, one can always find a product state such that the rank-one POVMs in the optimal measurement for distinguishing the resulting states is separable. Therefore, a two-qubit unitary transformation and its noisy counterpart can be perfectly discriminated by local resources only. In addition, the LOCC protocol can also be exploited to estimate the fraction of noise existing in a noisy channel. We present examples to demonstrate the LOCC protocol for channel identification and noise estimation.

\section{Identification of quantum gates}
\label{section:problem}


For a two-qubit gate $U$, its noisy counterpart $\N_U$ corresponds to a general quantum channel, i.e. completely positive and trace-preserving map over quantum states. Let us consider its noisy counterpart with random noise in the following form,
\bea
\label{nc}
\N_{U}^{p} [ \cdot] = (1-p ) U (\cdot) U^{\dagger}  + p D (\cdot) \label{eq:noisychannel}
\eea
where $p\in[0,1]$ is the fraction of noise and $D$ denote the complete depolarization, $D(\rho) = \mathrm{I}/4$ for all two-qubit states $\rho$. In order to find if implementation is given by gate $U$, or not, we consider the discrimination between gate $U$ and its noisy channel in a standard form $\N_{U}^p$.

\begin{figure}[t]
\begin{center}
\includegraphics[width=3.5in]{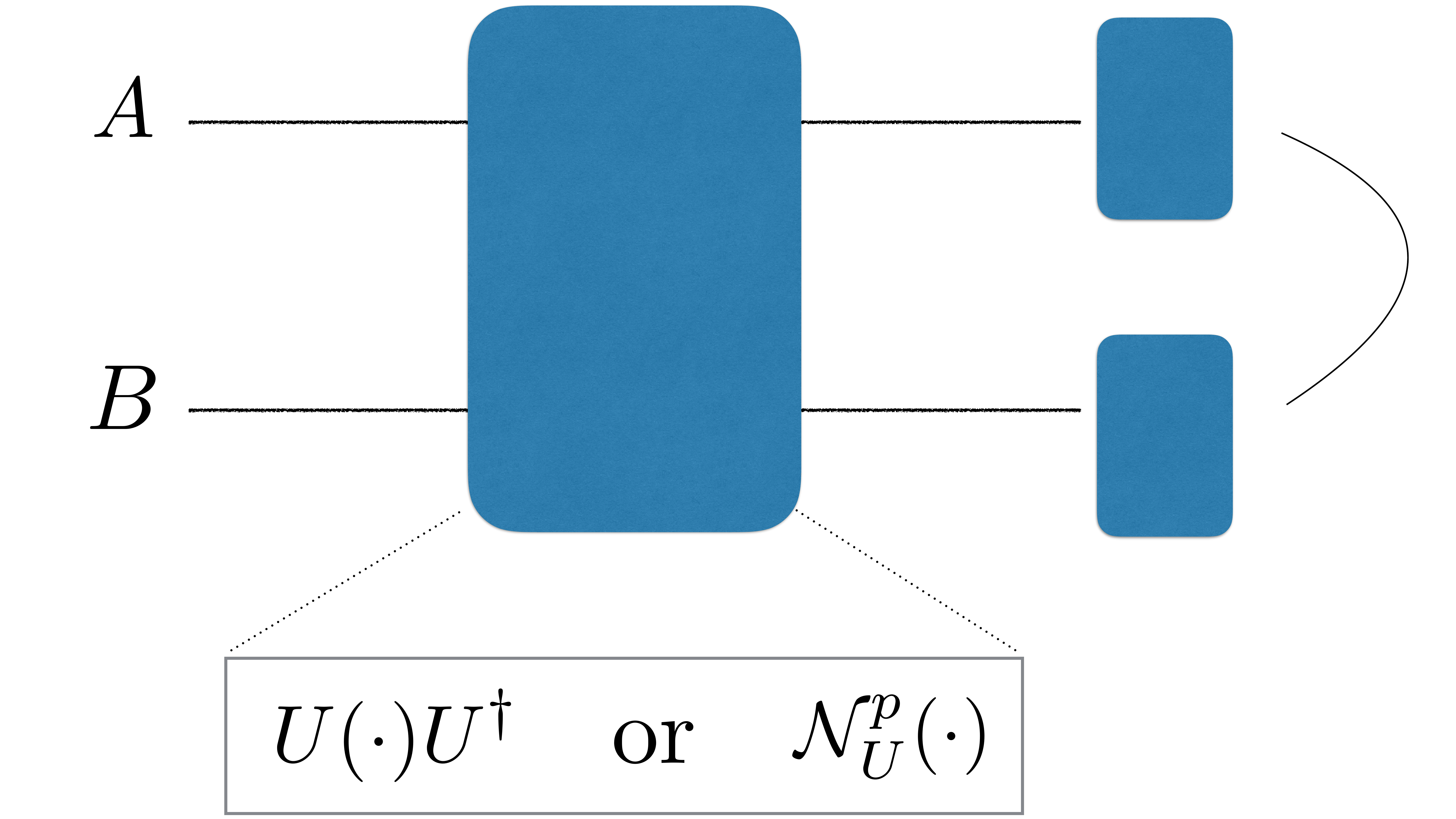}
\caption{ Once implementation of a two-buit gate $U$ is attempted, noise may occur with some probability, i.e., the ideal one $U$ or its noisy counterpart $\N_{U}^p$ in Eq. (\ref{eq:noisychannel}) appears with some probability. This can be identified as the problem of discriminating between two quantum operations, one unitary transformation and the other a general quantum channel. In an LOCC protocol, only local resources are exploited, such as the preparation of separable states $\rho_{AB} = \rho_A \otimes \rho_B$ in the beginning and local measurements wired by classical communication. A global strategy would incorporate preparation of an entangled state and collective measurement.  }\label{fig:fig1}
\end{center}
\end{figure}

To describe the channel identification, suppose that there is a box in which two operations a two-qubit gate $U$ and its noisy operation $\N_{U}^p$ are applied with  {\it a priori} probabilities $1-q$ and $ q$, respectively. The goal is to find if the box applies the unitary transformation or the noisy one in the one-shot scenario. One can prepare an input state and get measurement outcomes, from which the conclusion is then made. This is a typical scenario one has in a laboratory: preparing experimental apparatus to perform an operation on quantum states, one wants to know if the experimental setup corresponds to a desired operation or a noisy counterpart. 

Let $X$ denote the operation prepared in the box, $X=U,~\N_{U}^p$. We introduce an oracle function $f$ which knows the preparation of an operation in the box with certainty: it produces outcomes $0$ for $X=U$ and $1$ for $X=\N_{U}^p$. Then, we denote by $T$ an one-shot strategy to determine the value of the unknown variable $X$. Further restrictions can be made on the strategy $T$ such as LOCC when state preparation and measurement are performed locally together with classical communication. Channel identification can be described as the problem of finding a strategy $T$ which maximizes the probability that $T(X) = f(X)$:
\bea
p_{\g}   & = & \max_{T} \sum_{X =U, \N_{U}^p } \mathrm{Pr} [ T(X) = f(X)  ] \nonumber \\
 & = & \max_{T} ~ (1-q)~ \mathrm{Pr} [   T(X)  =0 |  f(X) =0  ] + q ~ \mathrm{Pr} [   T(X)  =1|  f(X) =1  ].\label{eq:chid}
\eea 
Recall that $1-q$ and $q$ are {\it a priori} probabilities: $X$ is given as a unitary transformation $U$ with probability $1-q$, and as its noisy counterpart $\N_{U}^p$ with probability $q$.

Suppose that the strategy applies global operations $\G$, together with $k$-dimensional ancillary system. For an input state $\rho$ to the box, there are two possibilities, $(\i_k \otimes U) \rho (\i_k \otimes U)^{\dagger}$ and $\i_k \otimes \N_{U}^p [\rho]$, which need to be discriminated by measurement. Note that ancillary systems $k>1$ are useful when quantum channel discrimination can be improved by exploiting entanglement between system and ancillas ~\cite{ref:mauro, ref:piani, ref:baedarek}. In this case, an optimal input state is entangled and the optimal measurement may also be a collective measurement on the global system, i.e., the system and ancillas.


Our goal here is to identify implementation of a two-qubit gate, which can generate entangled states. If entanglement between system and ancillas is to be exploited for the purpose, it means that another implementation of a two-qubit gate to generate entangled resources is required. Since we want to avoid the cost associated with any non-local resources, we restrict the consideration to cases where ancillas are not exploited and consequently entangled resources are not applied. When entangled resources in systems are not exploited, one can easily find that ancillary systems are of no use in this case to improve quantum channel discrimination. 

Without ancillary systems, discrimination of channels works by preparing an optimal state in such a way that the resulting states, after the channel use, are the most distinguishable. In this case, for global strategies $T=\G$, we have
\bea
p_{\g}^{(\G)}  & = & \max_{\rho} \max_{\Pi_0, \Pi_1}~ (1-q)\tr[ U \rho U^{\dagger} \Pi_0]   +  q \tr[ \N_{U}^{p} (\rho) \Pi_1]  \nonumber \\ 
& = &  \frac{1}{2} + \frac{1}{2} \max_{\rho}  \|  (1-q) U \rho U^{\dagger} - q \N_{U}^p (\rho) \|_1 \label{eq:pg1}
\eea
where $\Pi_0$ and $\Pi_1$ are POVMs such that $\Pi_0+ \Pi_1 =\mathrm{I}$, and $\| \cdot \|_1$ means the $L_1$-norm. This is a well-known result in minimum-error state discrimination between two quantum states ~\cite{ref:helstrom}, see also related reviews ~\cite{ref:r1} ~\cite{ref:r2} ~\cite{ref:r3} ~\cite{ref:r4} ~\cite{ref:r5}. Note that, in the above, the optimization is performed over all two-qubit states and then is achieved by a pure state. It is straightforward to compute the guessing probability, as follows,
\bea
p_{\g}^{(\G)} & = & \frac{1}{2} \left(1 +  \frac{3}{4}pq + \left| 1-2q + \frac{3}{4}pq \right| \right) \label{eq:pg2}
\eea
When $1-2q + \frac{3}{4}pq <0$, we have $p_{\g}^{(\G)} =q$ where the measurement contains only a single POVM, the identity. This corresponds to the case where no measurement is actually applied. The optimal strategy is to guess the noisy channel $\N_{U}^p$ all the time without measurement. When $1-2q + \frac{3}{4}pq \geq 0$, both POVM elements are non-zero, and the guessing probability is given by $p_{\g}^{(\G)} = 1-q +3pq/4$. The guessing probability is obtained with an optimal input state $|\psi\rangle$ that may be entangled and also by applying collective measurement, i.e. measurement in entangled basis. In what follows, we show that local resources can indeed attain optimal discrimination shown in Eq. (\ref{eq:pg2}) and also present the LOCC protocol.

\section{Optimal discrimination with local resources}
\label{section:optimaldiscrimination}

We here consider the strategy $T =\mathrm{LOCC}$ for optimal discrimination between the ideal one $U$ and its noisy counterpart $\N_{U}^{p}$ appearing with probability $1-q$ and $q$, respectively.  The so-called LOCC norm has been introduced in Ref. ~\cite{ref:winter} in an operational way,
\bea
\| X \|_{\mathrm{LOCC}} = \sup_{\mathcal{M}\in \mathrm{ LOCC}} \| \mathcal{M} (X) \|_1 \nonumber
\eea
where $\mathcal{M}$ denotes a set of POVMs or quantum instruments associated to LOCC. In terms of the LOCC norm, the guessing probability in Eq. (\ref{eq:chid}) can be found as follows,
\bea
p_{\g}^{(\mathrm{LOCC})}  & = & \max_{T =\mathrm{LOCC}} ~ (1-q)~ \mathrm{Pr} [  f (T(X)) =0 |  f(X) =0  ] + q ~ \mathrm{Pr} [  f (T(X)) =1|  f(X) =1  ]. \nonumber\\
& = & \frac{1}{2} + \frac{1}{2} \max_{\rho\in \mathrm{SEP}}  \|  (1-q) ~U \rho U^{\dagger} - q ~\N_{U}^p (\rho) \|_{\mathrm{LOCC}}.  \label{eq:pglocc}
\eea

In the above, the LOCC norm $\| \cdot \|_{\mathrm{LOCC}}$ can be achieved by LOCC protocol to discriminate between two resulting states $U \rho U^{\dagger}$ and $\N_{U}^p (\rho)$ appearing with probabilities $1-q$ and $q$ respectively. This corresponds to local measurements wired by classical communication, see Fig. (\ref{fig:fig1}). Then, the LOCC norm is to be maximized over all separable states, since only a separable state can be prepared in an LOCC protocol. Since pure states are extremal, it suffices to consider pure states in the optimization, i.e., $\rho = | \psi\rangle \langle \psi |  \otimes |\phi \rangle \langle \phi| $ for some states $|\psi\rangle\in \H_A$ and $|\phi\rangle \in \H_B$. In the following, we show that $p_{\g}^{(G)} =p_{\g}^{(\mathrm{LOCC})}$, i.e., in Eq. (\ref{eq:pglocc}) we have $\| \cdot \|_{\mathrm{LOCC}} = \|  \cdot \|_1$ and the maximization can be achieved by a product state.

\subsection{LOCC discrimination} 



To characterize quantum state discrimination with an LOCC protocol, we present the following theorem that leads to simplification of the analysis.

\begin{theorem}
\label{thm:1}
For bipartite quantum states $\{q_i ,\rho_i \}_{i=1}^2$, where $\rho_i\in S(\H_A\otimes \H_B)$ and $S$ denotes the set of quantum states on a Hilbert space, let $\{ \Pi_i\}_{i=1}^2$ denote POVMs for optimal discrimination. We also denote normalized POVMs by $\widetilde{\Pi_i} = \Pi_i / \tr[\Pi_i]$ for $i=1,2$, which thus correspond to quantum states. It holds that $\|q_1 \rho_1 - q_2 \rho_2 \|_{\mathrm{LOCC}} = \|q_1 \rho_1 - q_2 \rho_2 \|_1$ if and only if two normalized POVMs $\{\widetilde{\Pi_i} \}_{i=1}^2$ can be perfectly discriminated by an LOCC protocol. 
\end{theorem}


In the problem of discrimination between $U$ and $\N_{U}^{p}$, the resulting states $U\rho U^{\dagger}$ and $\N_{U}^{p}[\rho]$ for an input state $\rho$ may be entangled. Theorem 1 shows that from discrimination between two normalized optimal POVMs, one can find if optimal discrimination between the states in Eq. (\ref{eq:pg2}) can be achieved by an LOCC protocol. Therefore, we are here concerned with perfect discrimination by LOCC between two normalized POVMs $\widetilde{\Pi_1}$ and $\widetilde{\Pi_2}$, which correspond to the optimal discrimination between $U\rho U^{\dagger}$ and $\N_{U}^{p}[\rho]$ for an input state $\rho$.


Before proceeding to the proof, we describe the feature of a general LOCC protocol on a shared state $\rho_{AB}$. Without loss of generality, we assume that Alice first begins a protocol, in which $\{K_{j}^{A} \}$ denote her Kraus operators, i.e. it holds that $\sum_{j} {K_{j}^{A}}^{ \dagger} K_{j}^{A} = \mathrm{I}_A$. Alice's local operation on a shared state is described by $\{ K_{j}^{A} \otimes \i^{B} \}$. Bob acknowledges Alice's measurement outcome, denoted by $k_1$, according to which he devises local operations described by Kraus operators $\{ L_{j | k_1 }^{B} \}$ such that $\sum_j {L_{j | k_1 }^{B}}^{\dagger} L_{j | k_1 }^{B} = \mathrm{I}_B$. Let $l_1$ be Bob's outcome in the first round, after which the resulting state is given by, up to normalization,
\bea
\rho_{AB} ~\mapsto ~ (\mathrm{I}\otimes L_{l_{1} | k_1}^{B} ) (K_{k_1}^{A} \otimes \mathrm{I}) ~\rho_{AB} ~(K_{k_1}^{A} \otimes \mathrm{I})^{\dagger} ( \mathrm{I} \otimes L_{l_{1} | k_1}^{B})^{\dagger}. \nonumber
\eea
Note that this happens with the following probability
\bea
p_1 = \tr[ (\mathrm{I} \otimes  {{L_{l_{1} | k_1}^{B}}^\dagger} L_{l_{1} | k_1}^{B}  ) ( {{K_{k_1}^{A}}^\dagger} K_{k_1}^{A} \otimes \mathrm{I} ) \rho_{AB}]. \nonumber
\eea
According to the outcomes $(k_1, l_1)$, Alice decides local operations to apply, denoted by $\{ K_{j | (k_1,l_1)}^{A} \}$, and obtains an outcome denoted by $k_2$, corresponding to which Bob performs local operations $\{ L_{j |k_2 (k_1,l_1)}^{B} \}$. Let $(k_2,l_2)$ denote the measurement outcome in the second round. After the $n$-th round, we write the outcomes as 
\bea
(\vec{k}_n ,\vec{l}_n) := (k_n,l_n)(k_{n-1}, l_{n-1})\cdots (k_1, l_1). \nonumber
\eea
One can assume that, without loss of generality, an LOCC protocol terminates on the Bob's side with finite $n$. 

Then, the $n$-th Kraus operators of Alice and Bob can be generally written as
\bea
K_{(\vec{k}_{n} ,\vec{l}_{n} ) }^{A} & = &   K_{k_n | (\vec{k}_{n-1} ,\vec{l}_{n-1}) }^{A} K_{k_{n-1} | (\vec{k}_{n-2} ,\vec{l}_{n-2}) }^{A}\cdots K_{k_1}^{A}    \nonumber \\
L_{(\vec{k}_{n} ,\vec{l}_{n} ) }^{B} & = & L^{B}_{l_{n} | k_{n} (\vec{k}_{n-1} ,\vec{l}_{n-1}) } L^{B}_{l_{n-1} | k_{n-1} (\vec{k}_{n-2} ,\vec{l}_{n-2}) }\cdots L_{l_1|k_1}^{B}  \nonumber 
\eea
In this way, the resulting Kraus operators of Alice and Bob $\{ E_{( \vec{k}_n ,\vec{l}_n  )}^{\mathrm{AB_{LOCC}}} \}_{(\vec{k}_n ,\vec{l}_n )}  $ of the $n$ rounds for measurement outcomes $(\vec{k}_n ,\vec{l}_n )$ are described by
\bea
E_{ ( \vec{k}_n ,\vec{l}_n  )}^{\mathrm{AB_{LOCC}}}  =   K_{ (\vec{k}_{n } ,\vec{l}_{n } ) }^{A} \otimes L^{B}_{  (\vec{k}_{n } ,\vec{l}_{n } ) } \label{eq:KLOCC}
\eea
such that $\sum_{ ( \vec{k}_n ,\vec{l}_n  ) }   {E_{ ( \vec{k}_n ,\vec{l}_n  )}^{\mathrm{AB_{LOCC}}} }^{\dagger} E_{ ( \vec{k}_n ,\vec{l}_n  )}^{\mathrm{AB_{LOCC}}} =\mathrm{I}_{AB}$. With this description of LOCC, the proof of the aforementioned theorem is presented below.

\begin{proof}
\begin{subequations}
\label{eq:thm1}

($\Leftarrow$) Suppose the two states $\widetilde{\Pi_1}$ and $\widetilde{\Pi_2}$ can be perfectly discriminated by some LOCC protocol. This means that for all sequences of outcomes of the LOCC protocol, $(\vec{k}_n,\vec{l}_n)$, one can conclusively rule out one of the two states being present. This implies that all sequences $\{(k_n,l_n)\}$ can be partitioned into two classes: $\{(\vec{s}_n,\vec{t}_n) \}$ and $\{(\vec{v}_n,\vec{w}_n)\}$  such that POVM elements corresponding to them $ {E_{ ( \vec{s}_n ,\vec{t}_n  )}^{\mathrm{AB_{LOCC}}} }^{\dagger} E_{ ( \vec{s}_n ,\vec{t}_n  )}^{\mathrm{AB_{LOCC}}} $ and $ {E_{ ( \vec{v}_n ,\vec{w}_n  )}^{\mathrm{AB_{LOCC}}} }^{\dagger} E_{ ( \vec{v}_n ,\vec{w}_n  )}^{\mathrm{AB_{LOCC}}}$ satisfy the following
\bea
\tr[\widetilde{\Pi_2} ~ {E_{ ( \vec{s}_n ,\vec{t}_n  )}^{\mathrm{AB_{LOCC}}} }^{\dagger} E_{ ( \vec{s}_n ,\vec{t}_n  )}^{\mathrm{AB_{LOCC}}} ] =0~~\mathrm{and}~~\tr[\widetilde{\Pi_1}  ~{E_{ ( \vec{v}_n ,\vec{w}_n  )}^{\mathrm{AB_{LOCC}}} }^{\dagger} E_{ ( \vec{v}_n ,\vec{w}_n  )}^{\mathrm{AB_{LOCC}}}] = 0. \nonumber 
\eea
Since POVM elements of the LOCC protocol is complete, we have that
\bea
\sum_{(\vec{s}_n,\vec{t}_n)}{E_{ ( \vec{s}_n ,\vec{t}_n  )}^{\mathrm{AB_{LOCC}}} }^{\dagger} E_{ ( \vec{s}_n ,\vec{t}_n  )}^{\mathrm{AB_{LOCC}}} = \Pi_1 ~~\mathrm{and}~~ \sum_{(\vec{v}_n,\vec{w}_n)}{E_{ ( \vec{v}_n ,\vec{w}_n  )}^{\mathrm{AB_{LOCC}}} }^{\dagger} {E_{ ( \vec{v}_n ,\vec{w}_n  )}^{\mathrm{AB_{LOCC}}}} =  \Pi_2  \nonumber
\eea
This shows that the LOCC protocol implements the corresponding POVM: $\{ \Pi_1, \Pi_2\}$, and hence, the LOCC protocol is optimal to discriminate between states $\{q_i,\rho_i \}_{i=1}^{2}$. 

$(\Rightarrow)$ Conversely, we assume that Alice and Bob can implement the optimal discrimination for states $\{q_i, \rho_i \}_{i=1}^2$ by an LOCC protocol. This implies that all $(\vec{k}_n,\vec{l}_n)$ can be partitioned into two classes $\{(\vec{s}_n,\vec{t}_n) \}$ and $\{(\vec{v}_n,\vec{w}_n)\}$ such that the given states are optimally discriminated by the POVM elements in the following
\bea
\Pi'_1 \equiv \sum_{(\vec{s}_n,\vec{t}_n)}{E_{ ( \vec{s}_n ,\vec{t}_n  )}^{\mathrm{AB_{LOCC}}} }^{\dagger} E_{ ( \vec{s}_n ,\vec{t}_n  )}^{\mathrm{AB_{LOCC}}} ~~\mathrm{and}~~ \Pi'_2 \equiv \sum_{(\vec{v}_n,\vec{w}_n)}{E_{ ( \vec{v}_n ,\vec{w}_n  )}^{\mathrm{AB_{LOCC}}} }^{\dagger} E_{ ( \vec{v}_n ,\vec{w}_n  )}^{\mathrm{AB_{LOCC}}}. \nonumber
\eea
In Ref. ~\cite{ref:helstrom} it is shown that for two-state discrimination, POVM elements are unique, by which we have that $\Pi'_1= \Pi_1$ and $\Pi'_2 = \Pi_2$. Now note that the POVM elements $\Pi_1$ and $\Pi_2$, for two state discrimination, are projectors, hence $\Pi_1 \Pi_2 = 0$. This immediately implies that the states $\widetilde{\Pi}_1$ and $\widetilde{\Pi}_2$ can be perfectly discriminated by the same LOCC protocol.
\end{subequations}
\end{proof}

\subsection{LOCC preparation}

To exclude all possible non-local resources, we require not only measurement to be restricted to LOCC but also that our input state $\rho_{AB}$ be a separable state, see Fig. (\ref{fig:fig1}). Since pure states are extremal, it suffices to consider a product state, denoted by $\ket{\psi}$. That is, we aim to find a product state $\ket{\psi}$ such that normalized POVMs $\widetilde{\Pi}_1$ and $\widetilde{\Pi}_2$, given as optimal measurement for discrimination between the resulting states $U |\psi\rangle \langle \psi | U^{\dagger}$ and $\N_{U}^{p}[ |\psi\rangle\langle \psi | ]$, can be perfectly discriminated by LOCC. In fact, using the result of two-state discrimination \cite{ref:helstrom}, optimal POVMs can be found as follows,
\bea
\widetilde{\Pi}_1 = U |\psi\rangle \langle \psi | U^{\dagger}~~\mathrm{and}~~ \widetilde{\Pi}_2 = \frac{1}{3} (\mathrm{I} - U |\psi\rangle \langle \psi | U^{\dagger}) \label{eq:twopi}
\eea
In the following, we show the necessary and sufficient condition that $\widetilde{\Pi}_1$ and $\widetilde{\Pi}_2$ can be perfectly discriminated by LOCC. 


\begin{lem}[~\cite{ref:chit}]
\label{thm:theory:1}
For a two-qubit gate $U$, one can perfectly discriminate between the states $\widetilde{\Pi}_1$ and $\widetilde{\Pi}_2$ in Eq. (\ref{eq:twopi}) by LOCC if and only if $U \ket{\psi}$ is a product state. Then, the perfect discrimination can be obtained by a one-way LOCC protocol.
\end{lem}

\begin{proof}

($\Leftarrow$) Suppose that $U \ket{\psi}$ be a product state, denoted by $U \ket{\psi}= \ket{c}\ket{d},$ where $\ket{c} \in \mathcal{H}_A$ and $\ket{d} \in \mathcal{H}_B$. Then, the other POVM has a decomposition as follows, 
\begin{align}
\label{eq:theory:thm1:1}
\widetilde{\Pi}_2 = \frac{1}{3} \left( \proj{c_\perp,d} + \proj{c,d_\perp} + \proj{c_\perp,d_\perp} \right), 
\end{align}
where $\braket{c}{c_\perp} = \braket{d}{d_\perp} = 0$. The LOCC protocol for perfect discrimination between $\widetilde{\Pi}_1$ and $\widetilde{\Pi}_2$ is straightforward. Alice applies measurement in the orthonormal basis $\{ \ket{c}, \ket{c_\perp} \}$ and Bob does also in the orthonormal basis $\{ \ket{d}, \ket{d_\perp} \}$. Then if Alice obtains the outcome $\ket{c}$ and Bob the outcome $\ket{d}$, they declare that state $\widetilde{\Pi}_1$ is shared. Otherwise, they conclude state $\widetilde{\Pi}_2$. In this way, two parties can perfectly discriminated between two state $\widetilde{\Pi}_1$ and $\widetilde{\Pi}_2$. \smallbreak

($\Rightarrow$) Conversely, suppose that states $\widetilde{\Pi}_1$ and $\widetilde{\Pi}_2$ can be perfectly discriminated by a one-way LOCC protocol. Let Alice start the protocol, and $K_A$ denotes one of the Kraus operators of Alice's measurement in the one-way protocol for perfect discrimination. Consequently, the post-measurement states are given by,
\bea
(K_A\otimes \mathrm{I})\widetilde{\Pi}_1 (K_A\otimes \mathrm{I})^{\dagger}~~\mathrm{and}~~(K_A\otimes \mathrm{I})\widetilde{\Pi}_2 (K_A\otimes \mathrm{I})^{\dagger}. \label{eq:cont}
\eea
Since a Kraus operator $K_A$ on Alice's side leads to perfect discrimination, the post-measurement states in the above are orthogonal, i.e.,
\bea
\tr[(K_A\otimes \mathrm{I})\widetilde{\Pi}_1 (K_A\otimes \mathrm{I})^{\dagger} (K_A\otimes \mathrm{I})\widetilde{\Pi}_2 (K_A\otimes \mathrm{I})^{\dagger}] =0\nonumber
\eea
Rewriting the equation in the above, one can find that $  \widetilde{\Pi}_1  (K_{A}^{\dagger} K_A \otimes \mathrm{I}) \widetilde{\Pi}_2=0$. Let $\{  |\phi_1\rangle, |\phi_2\rangle, |\phi_3\rangle  \}$ be an orthonormal basis for the support of $\widetilde{\Pi}_2$. From Eq. (\ref{eq:twopi}) it follows that,

\bea
\tr_A [ K_{A}^{\dagger} K_A ( \tr_B [  |\phi_j\rangle\langle \psi| U^{\dagger} ] ) ]  =0,~~\forall j\in \{1,2,3\}. \label{eq:list}
\eea
That is, measurement ${K_A}^\dag K_A $ is orthogonal to the reduced operator $\tr_B [  |\phi_j\rangle\langle \psi| U^{\dagger}]$ for all $j=1,2,3$. Let $U \ket{\psi} $ have the following Schmidt decomposition.
\bea
U \ket{\psi} = \mu \ket{c}\ket{d} + \sqrt{1 - \mu^2} \ket{c_\perp}\ket{d_\perp} \label{eq:assume}
\eea
Suppose $\mu \in (0,1)$ for which the state $U | \psi \rangle$ in the above is entangled. Since $\widetilde{\Pi}_1$ and $\widetilde{\Pi}_2$ are orthogonal, one can find that the states $\{|\phi_{j}\rangle \}_{j=1}^3$ in the support of $\widetilde{\Pi}_2$ are written as follows,
\bea
|\phi_1 \rangle = - \sqrt{1-\mu^2} |c\rangle |d\rangle + \mu \ket{c_\perp} \ket{d_\perp},~ \ket{\phi_2} =  \ket{c} \ket{d_\perp},~ \mathrm{and}~\ket{\phi_3} =\ket{c_\perp}\ket{d} \nonumber 
\eea
for $|c\rangle \in\H_A$ and $|d\rangle \in \H_B$ under the assumption that $\mu\in (0,1)$. Consequently, from Eq. (\ref{eq:list}) the resulting operators $ \tr_B \left(\ketbra{\phi_j}{\psi}U^\dag  \right)$ for $j=1,2,3$ on the Alice side, that are orthogonal to $ K_{A}^{\dagger} K_A $, can be obtained as follows,
\bea
\mu \sqrt{1-\mu^2} \left( \proj{c_\perp} - \proj{c} \right), ~ \sqrt{1-\mu^2} \ketbra{c_\perp}{c}, ~\mathrm{and}~ \mu \ketbra{c}{c_\perp}. \label{eq:condition}
\eea
Then, we have $K_{A}^{\dagger} K_A \propto | c\rangle \langle c| + |c_{\perp}\rangle \langle c_{\perp}| =  \mathrm{I}$, that is, the measurement corresponds to an identity $\mathrm{I}$. This leads to the contradiction to the assumption that Alice's measurement can make two states in Eq. (\ref{eq:cont}) perfectly distinguishable, since the measurement is given by $\kappa\mathrm{I}$ for some $\kappa\in (0,1]$. Therefore, the state in Eq. (\ref{eq:assume}) is not entangled, i.e., we have $\mu =0$ or $\mu=1$, so that Alice's measurement $K_{A}^{\dagger} K_A$ can lead to perfect discrimination. We have shown that the state in Eq. (\ref{eq:assume}) is a product state. 
\end{proof}

\subsection{LOCC protocol for detecting noisy two-qubit gates}

We can now identify the problem of finding a product { two-qubit} state $|\psi\rangle$ such that the resulting state $U| \psi \rangle$ also remains a product state. From Lemma \ref{thm:theory:1}, this implies that two states in Eq. (\ref{eq:twopi}), optimal POVMs for discriminating between a two-qubit gate $U$ and its noisy counterpart $\N_{U}^{p}$, can be perfectly discriminated by LOCC. From Theorem \ref{thm:1}, therefore, two operations $U$ and $\N_{U}^{p}$ can be perfectly discriminated by LOCC. 

\begin{theorem}
\label{thm:theory:2}
For any two-qubit gate $U$, there exists a product state $\ket{\psi}$ such that is resulting state $U|\psi\rangle $ is also a product state. 
\end{theorem}

\begin{proof}
A two-qubit gate has a canonical decomposition as follows ~\cite{ref:kraus} ~\cite{ref:navin}, 
\begin{equation}
\label{U}
U = \left(  U_A \otimes U_B \right) U_d \left(V_A \otimes V_B \right),
\end{equation}
with $U_A$, $V_A$, $U_B$, and $V_B$ local unitary transformations and $U_d$ an entangling unitary transformation, 
\bea
&& U_d = \sum_{j=1}^{4} e^{i \lambda_j} \proj{\Phi_j} ~~ \mathrm{where} \label{eq:dunitary} \\
&& \ket{\Phi_1} = \dfrac{ 1}{\sqrt{2}} (\ket{00} + \ket{11} ),~ \ket{\Phi_2} =  \dfrac{1 }{\sqrt{2}} ( \ket{00} - \ket{11}), \nonumber \\
&&      \ket{\Phi_3} =  \dfrac{1 }{\sqrt{2}} ( \ket{01} - \ket{10}),~ \ket{\Phi_4} =  \dfrac{1}{\sqrt{2}} ( \ket{01} + \ket{10} ). \nonumber
\eea
In the following, we show that one can find a product state that is also a product state after an entangling gate. Then, we extend it to arbitrary two-qubit gates. 

A two-qubit state can be written in the basis in the above, 
\bea
\ket{\psi} = \sum_{j=1}^{4} \alpha_j \ket{\Phi_j} \nonumber
\eea
which is a product state if and only if 
\bea
{\alpha}_1^2 - {\alpha}_2^2 + {\alpha}_3^2 - {\alpha}_4^2 = 0. \label{eq:p}
\eea 
After applying an entangling gate in Eq. (\ref{eq:dunitary}), the resulting state $U|\psi\rangle $ is separable if and only if 
\bea
\left(e^{ i \lambda_1} {\alpha}_1\right)^2 - \left(e^{ i \lambda_2} {\alpha}_2\right)^2 + \left(e^{ i \lambda_3} {\alpha}_3\right)^2 - \left(e^{ i \lambda_4} {\alpha}_4\right)^2 = 0 \label{eq:up}
\eea
One aims to find a vector $\vec{ \alpha } = (\alpha_{1}^2, \alpha_{2}^2, \alpha_{3}^2, \alpha_{4}^2)$ such that the conditions of a product state in Eqs. (\ref{eq:p}) and (\ref{eq:up}) are satisfied. Let us rewrite Eqs. (\ref{eq:p}) and (\ref{eq:up}) as conditions as follows, 
\bea
\vec{t} \cdot \vec{\alpha}  = 0,~ \vec{ u_{re}} \cdot \vec{\alpha} = 0, ~\mathrm{and}~\vec{ u_{im}} \cdot \vec{\alpha} =0 \label{eq:ortho}
\eea
where three vectors are defined as $\vec{t} = (1,-1,1,-1)^T$, $u_{re} = (  \cos 2\lambda_1 , - \cos 2\lambda_2, \cos 2\lambda_3, - \cos 2\lambda_4)^T$, and $u_{im} = (  \sin 2\lambda_1 , - \sin 2\lambda_2, \sin 2\lambda_3, - \sin 2\lambda_4)^T$. These vectors define a subspace, denoted by $S_{U}$ {  in} $\mathbb{R}^4$, whose dimension is less than or equal to three, i.e., $\mathrm{dim} S_U {\color{red}\le} 3$. Hence, its orthogonal complement subspace $S_{U}^{\perp}$ is not a null space, i.e., $\mathrm{dim} S_{U}^{\perp}\ge 1$. Since $\vec{\alpha}$ is orthogonal to the subspace $S_{U}$, one can always find $\vec{\alpha} \in S_{U}^{\perp}$ that satisfies the conditions in Eq. (\ref{eq:ortho}). This shows the existence of a product state $|\psi\rangle$ that remains separable after application of an entangling unitary transformation. 

We now extend the result to arbitrary two-qubit gates. From the results shown so far, for an entangling unitary gate $U_d$ one can always find a product state $\ket{\psi}=  |a\rangle |b \rangle$ for some $|a\rangle\in\H_A$ and $|b\rangle \in \H_B$ such that such that $U_d \ket{a}\ket{b}$ is a product state, denoted by  $U_d \ket{a}\ket{b}= \ket{c}\ket{d}$. For an arbitrary two-qubit gate in Eq. (\ref{U}), choose $\ket{\psi} = \left(V_A^\dag \otimes V_B^\dag \right) \ket{a} \ket{b}$ so that $U \ket{\psi} = \left( U_A \otimes U_B \right) \ket{c}\ket{d}$, which is also a product state.  
\end{proof}

In addition, an LOCC protocol to discriminate between two optimal POVM elements can be constructed as follows. Two optimal POVM  elements $\widetilde{\Pi}_1 = U |\psi\rangle \langle \psi | U $ and $\widetilde{\Pi}_2 = \left( \mathrm{1} - U \proj{\psi} U^\dag \right)/3$ can be decomposed 
\bea
\widetilde{\Pi}_1  & = &   (U_A \otimes U_B)| c,d\rangle \langle c,d | (U_{A} \otimes U_B)^{\dagger} \nonumber \\
\widetilde{\Pi}_2  & \propto &    \left( U_A \otimes U_B \right)   \Big( \proj{c_\perp,d} + \proj{c,d_\perp}  +  \proj{c_\perp,d_\perp} \Big) \left(U_A^\dag \otimes U_B^\dag\right).  \nonumber
\eea
{  To distinguish between the POVM elements the following measurements are then performed during the LOCC protocol:} $\left\{ U_A\ket{c}, U_A \ket{c_\perp} \right\}$ on the { Alice'} side and $\left\{ U_B\ket{d}, U_B \ket{d_\perp} \right\}$ on the { Bob's} side, and {  they} communicate the measurement outcomes. {  When Alice and Bob get their respective outcomes as $U_A\ket{c}$ and $U_B\ket{d}$, they conclude that POVM element they had was $\widetilde{\Pi}_1$. For any other choice of outcomes, they conclude the POVM element $\widetilde{\Pi}_2$.} In this way, they can discriminate between two POVM elements by LOCC. Thus, the protocol can discriminate between a gate $U$ and its noisy one $\N_{U}^{p}$ in an optimal way.

\section{Examples}
\label{sec:ex}

We have shown that a two-qubit gate $U$ and its noisy counterpart $\N_{U}^{p}$ can be optimally discriminated by LOCC. The LOCC protocol can be constructed if there exists a product state that remains separable after application of a two-qubit gate $U$. In the following, we present examples of two-qubit gates that are often applied in a realistic scenario and show how the protocol works in practice.

\subsection{The Controlled-NOT gate}

The controlled-NOT (CNOT) gate is one of the standard components in finite sets of universal quantum gates. Let $U_{\mathrm{c}}$ denote the CNOT gate. {  It implements the following transformation on the computational basis states} $|a\rangle |b\rangle$ { where} $a,b\in \{ 0,1 \}:$ $U_{\mathrm{c}} | a\rangle |b\rangle = |a\rangle |b \oplus a\rangle$, where $\oplus$ { denotes} the addition modulo $2$. The results in Sec. \ref{section:optimaldiscrimination} has shown that the LOCC protocol can be constructed by finding a product { state} that remains separable under application of the CNOT gate. Then, local measurements are applied, and the conclusion of the discrimination between $U_c$ and $\N_{U_c}^{p}$ can be made by communication with the measurement outcomes. 

One can easily find that state $\ket{\psi}= \ket{0}\ket{0}$ remains the same under the CNOT gate, i.e. $U_c|00\rangle = |00\rangle$. POVMs for the optimal discrimination are obtained as follows,
\bea
{ \Pi_1} & = & |0  \rangle_{A } \langle 0 | \otimes |0  \rangle_{B } \langle 0 | \nonumber  \\
{ \Pi_2} & = &  \mathrm{I}_A \otimes \mathrm{I}_B -  |0  \rangle_{A } \langle 0 | \otimes |0  \rangle_{B } \langle 0 | . \label{eq:optcnot} 
\eea
Note that the optimal measurements are in the computational basis. 

As it is shown in the previous section, the LOCC protocol of discrimination between two POVMs $\widetilde{\Pi}_1$ and $\widetilde{\Pi}_2$ works for discrimination between $U_{\mathrm{c}}$ and $\N_{U_{\mathrm{c}}}^p$. Let $\mathrm{LOCC_c}$ denote the protocol, which works as follows. Alice and Bob apply measurements in the computational basis, $\{ \ket{0}_{A}, \ket{1}_{A} \}$ and $\{ \ket{0}_{B}, \ket{1}_{B} \}$, and communicate their measurement outcome. If they find the outcome $00$, they conclude the gate operation is noiseless, i.e., the unitary gate $U_{\mathrm{c}}$. Otherwise, for outcomes $01$, $10$, and $11$, they conclude that there exists noise, i.e., $\N_{U_{\mathrm{c}}}^p$ is performed. 

Let us illustrate how the LOCC protocol can be used to find if a gate operation is unitary or noisy in a practical realization, see Fig. (\ref{fig:fig1}).We consider a scenario that a gate operation is unitary $U_{\mathrm{c}}$ and noisy $\N_{U_{\mathrm{c}}}^p$ with probabilities $1-q$ and $q$, respectively. Recall the oracle function $f$ that gives outcomes $0$ if $X=U_c$ and $1$ if $X=\N_{U_c}^p$. Let $T = \mathrm{LOCC_c}$ denote the LOCC strategy. Then, we compute 
\bea
p_{\g}^{(\mathrm{LOCC})} = (1-q ) \mathrm{Pr} [  \mathrm{LOCC_c }(X)  =0 | f(X) = 0] +  q \mathrm{Pr} [  \mathrm{LOCC_c}(X)  =1 | f(X) = 1]\nonumber
\eea
where $ \mathrm{LOCC_c} (X)$ gives conclusion $0$ if the protocol $\mathrm{LOCC_c}$ finds $X$ as a unitary gate and $1$ otherwise. As it is explained in the above, for measurement outcome $00$ on register $AB$, we have $ \mathrm{LOCC_c}(X) =0$ and, otherwise for $01$, $10$, and $11$, we have $ \mathrm{LOCC_c}(X) =1$. 

Suppose that an input state $|\psi\rangle =|00\rangle$ is prepared. If the CNOT gate is applied, the outcome is $00$ with certainty: $\mathrm{Pr} [  \mathrm{LOCC_c }(X)  =0 | f(X) = 0]=1$. When its noisy counterpart is applied, we have $\mathrm{Pr} [  \mathrm{LOCC_c }(X)  =1 | f(X) = 1] = 3p/4$ since outcomes $01$, $10$ and $11$ appear with probability $p/4$ respectively. Thus we have, $p_{\g}^{(\mathrm{LOCC})} = 1-q +\frac{3}{4} pq$, which is equal to $p_{\g}^{(G)}$ in Eq. (\ref{eq:pg2}). Thus, for the CNOT gate it is shown that $p_{\g}^{(\mathrm{LOCC})} = p_{\g}^{(\G)}$.

\subsection{SWAP Gate}
Let $U_s$ denote the SWAP gate, that works as $U_s \ket{a}_A\ket{b}_B = \ket{b}_A \ket{a}_B$ for qubit states $\{ |a\rangle, |b\rangle\}\in \{|0\rangle,|1\rangle \}$. For an input state $\ket{\psi}_{AB} = \ket{a}_A \ket{b}_B$, the resulting state $U_s \ket{\psi} = \ket{b}_A \ket{a}_A$ is separable. The measurement outcome $ba$ on $AB$ leads to the conclusion that $X=U_s$, and otherwise for $\bar{b}a$, $b\bar{a}$ and $\bar{b}\bar{a}$ where $\bar{x}=1-x ~(\mod 2)$, a noisy channel $X = \N_{U_s}^p$. One can find that $p_{\g}^{(\mathrm{LOCC})} = p_{\g}^{(G)}$ for the SWAP gate, too.


\section{How useful is LOCC for channel discrimination? }
\label{sec:discussion}

In the construction of the LOCC protocol to detect a noisy channel, we have assumed that the noisy counterpart is in the form in Eq. (\ref{eq:noisychannel}). This can be obtained by a depolarization protocol ~\cite{ref:dur} or the technique of twirling quantum channels ~\cite{ref:eme1, ref:eme2}. Then, it is shown that optimal discrimination between two-qubit unitary $U$ and $\N_{U}^p$ can be achieved by LOCC, i.e., we have
\bea
p_{\g}^{(\mathrm{LOCC})} = \frac{1}{2} \big(  1+ \frac{3}{4}pq + \big | 1-2q + \frac{3}{4}pq \big|  \big),  \label{eq:pglocc1}
\eea
for which the protocol can be constructed from the results. 

We remark the significance of the depolarization before proceeding to the LOCC discrimination. Suppose that that the depolarization is not yet applied, in particular, the noisy channel is given as 
\bea
 \N_{U_c} [\rho] = (1-p) U_c \rho U_{c}^\dag + p U' \rho U'^\dag, \label{benasque:noisy} 
\eea
where $U_c$ is the CNOT gate and $U' = U_c \left(U_A \otimes U_B \right)$ with $U_A = \proj{0} + i \proj{1}$ and $U_B = \proj{0} + i \proj{1}$. With global strategies, one can find that $p_{\g}^{(\G)} = (1/2 + 3p/8)$, when $q$ is set equal to $1/2$. This can be achieved with the input state $\ket{\phi^{+}} = (|00\rangle + |11\rangle)/\sqrt{2}$. 

However, the optimal discrimination between two operations $U_c$ and $\N_{U_c}$ cannot be achieved by LOCC, i.e. $p_{\g}^{(\mathrm{LOCC})} < p_{\g}^{(G)}$. To have that $p_{\g}^{(G)} = p_{\g}^{(\mathrm{LOCC})}$ we need the condition that the following holds true
\bea
\max_{\rho} \| U_c \rho U_{c}^{\dagger} - \N_{U_c}[\rho] \|_1 =  \max_{\rho \in \mathrm{SEP}} \| U_c \rho U_{c}^{\dagger} - \N_{U_c}[\rho] \|_{\mathrm{LOCC}}. \nonumber
\eea
One can find that the above can be rewritten equivalently as follows,
\bea
\max_{\rho} \| U_c \rho U_{c}^{\dagger} -  U' \rho U'^{\dagger} \|_1 =  \max_{\rho\in \mathrm{SEP}} \| U_c \rho U_{c}^{\dagger} -  U' \rho U'^{\dagger} \|_{\mathrm{LOCC}}, \nonumber
\eea
which, however, does not hold true in general as it is shown in Ref. ~\cite{ref:bae}. That is, for discrimination between $U$ and $U'$, there is a gap between LOCC and global strategies.

\subsection*{Estimation of the Noisy Parameter $p$}

We remark that Eq. (\ref{eq:pglocc1}) shows a relation between the guessing probability and the noise fraction $p$. Recall that $q$ is {\it a priori} probability. This means that, conversely, the noise fraction $p$ can be obtained by finding $p_{\g}^{(\mathrm{LOCC})}$ if $q$ is known. Note that the value of $p$ depends on the noisy channel $\N_{U}^p$ only, i.e., it does not depend on the value of the {\it a priori} probability $q$.  Since Alice and Bob do not have a noiseless unitary $U$ in their possession, they can choose $q$ to be equal to one, i.e., $q=1$. When $q=1$, then the term $1-2q + \frac{3}{4}pq $ becomes $\frac{3}{4}p-1$, which is always negative, i.e., $\frac{3}{4}p < 1$. Despite this, we can still use our derivation in Eq. (\ref{section:optimaldiscrimination}) to estimate $p$ using only local resources. In fact, Alice and Bob can use the same protocol as mentioned in Eq. (\ref{section:optimaldiscrimination}), but now, instead of using the protocol for a single-shot channel identification, they use multiple rounds to obtain the frequences $f_0$ for outcome $U_A\ket{c}$ and $U_B\ket{d}$ respectively, and $f_1$ for any other outcomes. From Eq. (\ref{nc}), it is easy to see that $p = \frac{4}{3}f_1$. Thus, they can estimate the value of $p$ using only local resources.


\section{ Conclusion} 
\label{sec:conclusion}

In this work, we have exploited quantum channel discrimination to detect a noisy channel in quantum gates. We have investigated discrimination between a two-qubit unitary transformation and its noisy counterpart with local resources. It is shown that for the noisy counterpart with the depolarization channel, LOCC suffice to achieve optimal discrimination of global strategies. The LOCC protocol is also illustrated with examples of the CNOT and the SWAP gates. Note that the protocol can be readily applied in practice to detect if a gate in a quantum circuit is noisy, or not. Moreover, the noise fraction can also be estimated by the protocol. Finally, our results can be generalized to arbitrary noisy channels with the channel twirling ~\cite{ref:dur, ref:eme1, ref:eme2}, where the depolarization of a noisy channel, however, contains applications entangled unitary transformations. It would be interesting to find if twirling of quantum channels can be replaced with an LOCC protocol, so that detection of noisy channels can be generally performed with local resources only.

\section{Acknowledgement}
 This work is supported by National Research Foundation of Korea (NRF2017R1E1A1A03069961), the ITRC(Information Technology Research Center) support program(IITP-2018-2018-0-01402) supervised by the IITP(Institute for Information \& communications Technology Promotion) and the KIST Institutional Program (Project No. 2E26680-18-P025).


\end{document}